\documentclass[letterpaper,11pt]{article}
\usepackage[T1]{fontenc}
\usepackage[utf8]{inputenc}
\usepackage{lmodern}
\usepackage{microtype}
\usepackage{complexity}
\usepackage{amsmath}
\usepackage{amssymb}
\usepackage{amsthm}
\usepackage{amsfonts}
\usepackage{latexsym}
\usepackage{bbm}
\usepackage{pdfpages}
\usepackage{algorithmicx}
\usepackage{algorithm}
\usepackage[top=1in,bottom=1in,left=1in,right=1in]{geometry}
\usepackage[textsize=tiny]{todonotes}
\usepackage{graphicx}
\usepackage{xspace}
\usepackage[unicode=true]{hyperref}
\usepackage[capitalise]{cleveref}
\usepackage{bm}
\usepackage{subcaption}
\usepackage{mathtools}

\usepackage{algorithm}
\usepackage[noend]{algpseudocode}

\newtheorem{theorem}{Theorem}
\newtheorem{lemma}[theorem]{Lemma}

\newtheorem{fact}[theorem]{Fact}

\newtheorem*{remark}{Remark}

\theoremstyle{definition}
\newtheorem{definition}[theorem]{Definition}

\newcommand{\OPT}{\mathrm{OPT}}

\newcommand{\opt}{\mathrm{opt}}
\newcommand{\dist}{\mathrm{dist}}
\newcommand{\cost}{\mathrm{cost}}
\newcommand{\grid}{\mathrm{grid}}
\newcommand{\pro}{\mathrm{pro}}
\newcommand{\len}{\mathrm{length}}
\newcommand{\eps}{\varepsilon}
\renewcommand{\epsilon}{\varepsilon}
\newcommand{\CE}{\mathrm{CE}}
\newcommand{\NCE}{\mathrm{NCE}}

\renewcommand{\D}{\mathrm{D}(\bm{a})}
\newcommand{\qt}{\mathrm{QT}}
\newcommand{\Q}{\qt(P,\bm{a})}
\newcommand{\cqt}{\mathrm{CQT}}
\newcommand{\CQ}{\cqt(P,\bm{a})}

\renewcommand{\AC}{A}

\newcommand{\KNW}{Kisfaludi-Bak, Nederlof, and Węgrzycki\xspace}

\renewcommand{\E}{\mathbb{E}}
\newcommand{\PD}{P}

\title{A Linear Time Gap-ETH-Tight Approximation Scheme \\for Euclidean TSP}

\date{}
\author{Tobias Mömke\thanks{University of Augsburg, Germany. Partially supported by DFG Grant 439522729 (Heisenberg-Grant) and DFG Grant 439637648 (Sachbeihilfe).} \and Hang Zhou\thanks{École Polytechnique, France. \url{https://www.normalesup.org/\~zhou}. Partially supported by the Hi!Paris Grant.}}

\begin{document}
\maketitle
\thispagestyle{empty}
\begin{abstract}
The Traveling Salesman Problem (TSP) in the $d$-dimensional Euclidean space is among the oldest and most famous NP-hard optimization problems.
In breakthrough works, Arora [J. ACM 1998] and Mitchell [SICOMP 1999] gave the first polynomial time approximation schemes.
To improve the running time, Rao and Smith [STOC 1998] gave a randomized $(1/\varepsilon)^{O(1/\varepsilon^{d-1})}\cdot n\log n$ time approximation scheme.
Bartal and Gottlieb [FOCS 2013] gave a randomized approximation scheme in $2^{(1/\varepsilon)^{O(d)}} n$ time, which is linear in $n$.
Recently, Kisfaludi-Bak, Nederlof, and Węgrzycki [FOCS 2021] gave a randomized approximation scheme in $2^{O(1/\varepsilon^{d-1})} n \log n$ time, achieving a Gap-ETH tight dependence on $\varepsilon$.
It is raised as a challenging open question by Kisfaludi-Bak, Nederlof, and Węgrzycki [FOCS 2021] whether a running time of $2^{O(1/\eps^{d-1})}n$ is achievable.
We answer their question positively by giving a randomized $2^{O(1/\varepsilon^{d-1})} n$ time approximation scheme for Euclidean TSP.
\end{abstract}

\newpage
\setcounter{page}{1} 

\section{Introduction}

The Traveling Salesman Problem (TSP) in the $d$-dimensional Euclidean space is among the oldest and most famous NP-hard optimization problems.
The problem is to find a shortest tour on a given set of $n$ points in $\mathbb{R}^d$.
The Gödel-prize-winning approximation schemes for this problem due to Arora~\cite{Aro98} and Mitchell~\cite{Mit99} are among the most prominent results in approximation algorithms.
They are classics in advanced algorithmic courses as well as in textbooks on approximation algorithms, optimization, and geometric algorithms~\cite{Vaz2010,williamson2011design,korte12,har2011geometric,narasimhan2007geometric}.
The methods of Arora~\cite{Aro98} and Mitchell~\cite{Mit99} have numerous applications; see, e.g., \cite{arora2003approximation} for a survey. These methods are gateways to a vast variety of problems in geometric optimization, see, e.g., \cite{KNW21} for a summary.

The randomized approximation schemes in \cite{Aro98} and in \cite{Mit99} have a running time of $n(\log n)^{O(1/\eps^{d-1})}$ and $n^{O(1/\eps)}$, respectively.\footnote{The algorithm in \cite{Mit99} works when $d=2$.} The most natural goal is to improve their running times  to be \emph{optimal}.
In the last 25 years only three such results were obtained:
\begin{enumerate}
\item Rao and Smith~\cite{RS98} used geometric spanners to achieve a $(1/\varepsilon)^{O(1/\varepsilon^{d-1})}\cdot n\log n$ time randomized approximation scheme.
\item Bartal and Gottlieb~\cite{BG13} gave a $2^{(1/\varepsilon)^{O(d)}} n$ time randomized approximation scheme in the real-RAM model with atomic floor or mod operations. The running time is truly linear in $n$, but the dependence on $\eps$ is not as good as the algorithm of Rao and Smith~\cite{RS98}.
\item Very recently, Kisfaludi-Bak, Nederlof, and Węgrzycki~\cite{KNW21} gave a $2^{O(1/\varepsilon^{d-1})} n \log n$ time randomized approximation scheme. They showed that this running time has a tight dependence on $\eps$ under the Gap-Exponential Time Hypothesis (Gap-ETH). However, the dependence on $n$ is not linear.
\end{enumerate}

Kisfaludi-Bak, Nederlof, and Węgrzycki~\cite{KNW21} raised the following open question:

\vspace{3mm}
\emph{``The ideal algorithm for Euclidean TSP would have a running time of $2^{O(1/\eps^{d-1})}n$, and it would be deterministic. However, achieving this running time with a randomized algorithm is already a challenging question.''}
\vspace{3mm}

In our work, we answer their question by giving a randomized algorithm with running time $2^{O(1/\eps^{d-1})}n$ for Euclidean TSP (\cref{thm:main}).

\begin{theorem}[main theorem]\label{thm:main}
    There is a randomized $(1+\eps)$-approximation scheme for the Euclidean TSP in $\mathbb{R}^d$ that runs in time $2^{O(1/\eps^{d-1})}n$ in the real-RAM model with atomic floor or mod operations.
\end{theorem}

\begin{remark}
The real-RAM model with atomic floor or mod operations is crucial to achieve a running time that is linear in $n$. This model was previously used in the open question raised by Rao and Smith~\cite{RS98} regarding a linear running time and in the linear time algorithm by Bartal and Gottlieb~\cite[Corollary I.2]{BG13}.
Indeed, in other models of computation, a linear running time is unachievable. For example, there is an $\Omega(n \log n)$ time lower bound for Euclidean TSP in the decision tree model~\cite{DKS97,RS98}.
\end{remark}

The running time in \cref{thm:main} has a tight dependence on $\eps$ under the Gap-ETH: if there is a $2^{o(1/\eps^{d-1})}\cdot \poly(n)$ time algorithm in the real-RAM model with atomic floor or mod operations, then there is a $2^{o(1/\eps^{d-1})}\cdot \poly(n)$ time algorithm in the decision tree model, which would refute Gap-ETH according to~\cite{KNW21}.

\subsection{Related Works}

Achieving an algorithm for TSP with a (conditionally) optimal running time is the goal in many contexts.
For example, De~Berg, Bodlaender, Kisfaludi-Bak, and Kolay~\cite{de2023eth} gave the best-to-date \emph{exact} algorithm for Euclidean TSP in time $2^{O(n^{1-1/d})}$ and a matching lower bound of $2^{\Omega(n^{1-1/d})}$ under the Exponential Time Hypothesis (ETH).
For example, for \emph{planar graphs}, Klein~\cite{klein2008linear} gave a $2^{O(1/\eps)} n$ time approximation scheme for planar TSP, and Marx~\cite{marx2007optimality} showed that the dependence on $\eps$ in Klein's algorithm is near-optimal under the ETH.

\subsection{Notations and Previous Techniques}

Let $P\subseteq \mathbb{R}^d$ denote a set of $n$ points.
By preprocessing the input instance, we may assume that $P\subseteq \{0,\dots,L\}^d$ for some integer $L=O(n/\eps)$ that is a power of 2.
This preprocessing step can be done in $O(n)$ time~\cite{BG13}.

For two points $u,v\in \mathbb{R}^d$, let $\dist(u,v)$ denote the Euclidean distance between $u$ and $v$.
For a segment $I$ in $\mathbb{R}^d$, let $\len(I)$ denote the length of $I$.
For an edge $e=(u,v)$, let $\cost(e):=\dist(u,v)$.
For a set of edges $S$, we define $\cost(S):=\sum_{e\in S}\cost(e)$.
For any path $\pi$ of points $x_1,x_2,\dots, x_m$ in $\mathbb{R}^d$ where $m\in \mathbb{N}$, define $\cost(\pi)=\sum_{i=1}^{m-1} \dist(x_i,x_{i+1})$.
We say that $\pi$ is a \emph{tour} if $x_1=x_m$.

Let $\OPT$ denote an optimal tour on $P$.
Let $\opt$ denote $\cost(\OPT)$.

\subsubsection{Dissection and Quadtree}
We follow the notations in \cite{KNW21}.
Let $d\geq 2$ be a constant.
We pick $a_1,\dots,a_d\in \{1,\dots,L\}$ independently and uniformly at random and define $\bm a := (a_1,\ldots, a_d)\in\{0,\dots,L\}^d$.
Consider the hypercube
\[C(\bm a):= \bigtimes\limits_{i=1}^d [-a_i+1/2, 2L-a_i+1/2].\]
Note that $C(\bm a)$ has side length $2L$ and each point from $P$ is contained in $C(\bm a)$ since $P\subseteq\{0,\dots,L\}^d$.

We define the {\it dissection} $D(\bm a)$ of $C(\bm a)$ to be a tree constructed recursively, where each vertex is associated with a hypercube in $\mathbb R^d$.
The root of the tree is associated with $C(\bm a)$.
Each non-leaf vertex of the tree that is associated with a hypercube $\bigtimes_{i=1}^d [l_i,u_i]$ has $2^d$ children with which we associate hypercubes $\bigtimes_{i=1}^d I_i$, where $I_i$ is either $[l_i,(l_i+u_i)/2]$ or $[(l_i+u_i)/2,u_i]$. Each leaf vertex of the tree is associated with a hypercube of unit length.
For each hypercube in $D(\bm a)$, the \emph{level} of its facets is the hop distance in $D(\bm a)$ between the hypercube and the root of $D(\bm a)$.
For a hyperplane $h$, we define the \emph{level} of $h$ to be the smallest level of any facet contained in $h$.

A \emph{quadtree} $\Q$ is obtained from $D(\bm a)$, by stopping the recursive partitioning as soon as the associated hypercube of a vertex contains at most one point from $P$.
Each hypercube associated with a vertex in the quadtree is called a {\it cell}.
We say that a cell $C$ is \emph{redundant} if it has a child that contains the same set of input points as the parent of $C$. A redundant path is a maximal ancestor-descendant path in the tree whose internal vertices are redundant.
The \emph{compressed quadtree} $\CQ$ is obtained from the quadtree by removing all the empty children of redundant cells and replacing the redundant paths with edges.
In the resulting tree some internal cells may have a single child; we call these \emph{compressed cells}.
It is well-known that compressed quadtrees have $O(n)$ vertices, see, e.g.,~\cite{BG13}.

\subsubsection{Approach of Arora\texorpdfstring{~\cite{Aro98}}{}}

\begin{lemma}[Arora's Patching Lemma \cite{Aro98}, see also~\cite{KNW21}]
\label{lem:Arora-patch}
Let $h$ be a hyperplane. Let $\pi$ be a tour.
Let $I(\pi,h)$ denote the set of intersections of $\pi$ with $h$.
Assume that $h$ does not contain any endpoints of segments that define $\pi$.
Let $T$ be a tree on the hyperplane $h$ that spans $I(\pi,h)$.
Then for any point $p$ in $T$, there exist line segments contained in $h$ whose toatl length is at most $O(\cost(T))$ and whose addition to $\pi$ changes it into a tour $\pi'$ that crosses $h$ at most twice and only at $p$.
\end{lemma}

Using \cref{lem:Arora-patch}, Arora proves a structure theorem, which states that there is a near-optimal tour that crosses each facet of a cell $O(1/\eps)$ times and only through portals.
While such a near-optimal tour does not necessarily exist for a fixed quadtree, a randomly shifted quadtree works with high probability.
Thus the algorithm first computes a randomly shifted quadtree, and then applies a dynamic program to compute the best solution satisfying the properties in the structure theorem.

\subsubsection{Approach of \KNW\texorpdfstring{~\cite{KNW21}}{}}
\label{sec:KNW}
To achieve a Gap-ETH tight running time, \KNW~\cite{KNW21} introduce a powerful technique of the \emph{sparsity-sensitive patching}: for a facet of a cell that is crossed by a tour at $1<k\leq O(1/\eps)$ crossings, they modify the tour by moving each crossing to the nearest portal from a set of $g(k)$ equidistant portals. Here $g(k)=\Theta(1/(\eps^2 k))$ is a granularity  parameter that depends on $k$.
However, applying the sparsity-sensitive patching has a condition: the number of crossings $k$ is such that $k\neq 1$. This condition is necessary: when $k=1$, \KNW~\cite{KNW21} provide an example showing that moving the single crossing to the closest portal \emph{does not} lead to a near-optimal solution.
To deal with single crossings, their solution is to make use of a spanner.
The $O(n\log n)$ running time in \cite{KNW21} is due to the computation of the spanner.
In our work, we use new ideas to deal with single crossings; see \cref{sec:overview}.

\subsubsection{Linear-time 2-approximation by Bartal and Gottlieb\texorpdfstring{~\cite{BG13}}{}
}
\label{sec:BG13}
The first step in our approach requires computing a 2-approximate solution.
This can be achieved in $O(n)$ time, by  setting $\eps=1$ in the main result of Bartal and Gottlieb~\cite{BG13}. Alternatively, we could also compute a minimum spanning tree in linear time using Chan's result~\cite{Cha08} and then double the edges in the spanning tree to obtain a TSP tour.
\begin{lemma}[Bartal and Gottlieb~\cite{BG13}]
\label{lem:2-approx}
    There is randomized 2-approximation algorithm for the TSP in $\mathbb{R}^d$ that runs in time $O(n)$ in the real-RAM model with atomic floor or mod operations.
\end{lemma}

\subsubsection{Portals}

\begin{definition}[grid, {\cite[Definition~II.4]{KNW21}}]
\label{def:grid}
Let $F$ be a $(d-1)$-dimensional hypercube.
Let $t$ be a positive integer.
We define $\grid(F,t)\subseteq\mathbb{R}^{d-1}$ to be an orthogonal lattice of $t$ points in $F$.
\end{definition}

As a consequence, if the hypercube has side length $l$, the minimum distance between any pair of points of $\grid(F, t)$ is $l/t^{1/(d-1)}$.
The following function $g(k)$ computes the number of portals to be placed on a facet, given the number $k$ of crossing points on that facet.
\begin{definition}[adaptation from \cite{KNW21}]
    \label{def:granularity}
    Let $r$ be a positive integer.
    For a positive integer $k$, let $q$ denote the smallest integer such that $q^{d-1} \geq (r/2)^{2d-2}/k$ and $q$ is an integer power of 2.
    Then we define $g(k):=q^{d-1}$.
\end{definition}

\subsection{Overview of Our Methods}
\label{sec:overview}
In this section, we give an  overview of our methods in the proof of \cref{thm:main}.

Our approach is based on the approach of \KNW~\cite{KNW21} and uses several new ideas.
To improve upon the running time from \cite{KNW21}, a main challenge is to deal with the case when a facet has a single crossing with a tour, see \cref{sec:KNW}.
That case is the reason why the running time in \cite{KNW21} requires a multiplicative factor of $\log n$.
In order to get rid of the $\log n$ factor, we deal with the single crossing case in a new way.

We start by computing a 2-approximate solution $T$ in linear time, using an algorithm of Bartal and Gottlieb~\cite{BG13}, see~\cref{sec:BG13}.
The solution $T$ helps us in dealing with single crossings.

\subsubsection{\texorpdfstring{$r$}{r}-basic Tours}
We define a set of portals for each facet $F$, consisting of the portals in \cite{KNW21}, plus possibly an additional point, which is a crossing between $T$ and $F$ if it exists.\footnote{When $T$ and $F$ has multiple crossings, an arbitrary crossing would work.}
See \cref{def:Z}.

\begin{definition}
\label{def:Z}
Let $r=O(1/\eps)$.
We say that a facet $F$ of a cell is \emph{$T$-crossing} if $T$ crosses $F$ through at least one point.
If $F$ is $T$-crossing, we let $x_F$ denote an \emph{arbitrary} crossing point between $F$ and $T$.
For a positive integer $k$, define
\[
    Z(F,k):=\begin{cases}
    \grid(F,g(k))\cup \{x_F\}, & \text{if $F$ is $T$-crossing;} \\
    \grid(F,g(k)), &\text{otherwise,}
    \end{cases}
\]
where the function $g(\cdot)$, which depends on $r$, is defined in \cref{def:granularity}, and the function $\grid(\cdot,\cdot)$ is defined in \cref{def:grid}.
\end{definition}

We consider tours that respect the portals in \cref{def:Z}. Such tours are called \emph{$r$-basic} tours (\cref{def:r-basic}).
An $r$-basic tour is analogous to an \emph{$r$-simple} tour in \cite{KNW21}.

\subsubsection{Structure Theorem}

We prove a structure theorem (\cref{thm:structure}), which shows the existence of an $r$-basic solution that is near optimal.

\begin{theorem}[Structure Theorem]
\label{thm:structure}
Let $P$ be a set of $n$ points in $\mathbb{R}^{d}$.
Let $\bm{a}\in \{1,\dots,L\}^d$ be a random shift.
For any large enough integer~$r$, there is an $r$-basic tour $\OPT'$ on $P\subseteq \mathbb{R}^d$ such that \[\mathbb{E}_{\bm{a}}[\cost(\OPT')]\leq (1+O(1/r))\cdot \opt.\]
\end{theorem}

The proof of the Structure Theorem is a main novelty in this paper.
First, we prove the Structure Theorem in 2 dimensions (\cref{sec:structure}), and then we generalize the proof to $d$ dimensions (\cref{sec:d-dimensions}).

In the rest of this section, we sketch high level ideas of the proof for 2 dimensions.

To begin with, when a side $F$ of a cell has a single crossing with $\OPT$ and when in addition $F$ is $T$-crossing (\cref{def:Z}), then we move the crossing between $F$ and $\OPT$ to $x_F$.
We observe that this connection cost is negligible (\cref{lem:single-crossing-and-T-crossing}).
Therefore, it suffices to consider the case when $F$ has a single crossing with $\OPT$ but has no crossing with $T$. We call such a side $F$ \emph{interesting}.

A crucial observation is that, when $F$ is interesting, we are able to afford moving the single crossing between $F$ and $\OPT$ to the closest portal (\cref{lem:interesting}).
To show \cref{lem:interesting}, for each edge $e\in\OPT$ that crosses an interesting side, we upper bound the connection cost for that crossing by an interportal distance, and we analyze the overall connection cost in \cref{lem:q(e)}.
\cref{lem:q(e)} a main technical contribution in this paper.


Now we sketch the proof of \cref{lem:q(e)}.
For an edge $e$ crossing an interesting side, we bound the connection cost for that crossing in an original way.
To begin with, we define \emph{badly cut} edges (\cref{def:badly_cut}).
This definition is inspired by Cohen-Addad~\cite{cohen2020approximation}.
We observe that it suffices to consider the case when $e$ is badly cut.
Our analysis for badly cut edges is completely different from \cite{cohen2020approximation}.
The non-trivial subcase is when $e$ is \emph{critical} (\cref{fig:critical}).
To analyze critical edges, our approach (\cref{lem:CE}) combines \cite{KNW21} with several new ingredients:
First, we observe that, in \cite{KNW21}, a single crossing is problematic because the \emph{proximity} of that crossing  may be arbitrarily large.
To overcome that difficulty, we show that for a \emph{critical} edge, the proximity of \emph{another} crossing on a \emph{perpendicular}  side is bounded.
Care is needed, because in \cite{KNW21}, the patching cost for a single crossing is undefined. We define the patching cost in this situation to be the interportal distance with small probability, by analyzing horizontal and vertical lines simultaneously.
This is in contrast to the analysis in \cite{KNW21} that is only on horizontal lines or only on vertical lines.

Our approach extends to $d$ dimensions: instead of considering another crossing on a perpendicular side, we consider another crossing on a perpendicular square. See \cref{sec:d-dimensions}.

\subsubsection{Algorithm}
From the Structure Theorem, in order to compute a near-optimal solution for Euclidean TSP, it suffices to compute a shortest $r$-basic solution using a dynamic program.
Hence our algorithm for Euclidean TSP (\cref{alg:main}).

\begin{algorithm}[H]
    \caption{Algorithm for Euclidean TSP}
    \label{alg:main}
    \begin{algorithmic}[1]
        \State Compute a 2-approximate solution $T$ according to \cref{lem:2-approx}
        \State Pick a random shift $\bm{a}$ and compute a compressed quadtree $\CQ$ from the dissection $\D$
        \For{each side $F$ of each cell $C$}
                \If{$F$ is $T$-crossing}
                    \State $x_F \gets$ an \emph{arbitrary} crossing point between $T$ and $F$
                \EndIf
        \EndFor
        \State Use a dynamic program to compute a shortest $r$-basic solution
    \end{algorithmic}
\end{algorithm}

The computation of the crossing points (Lines 3--5 of \cref{alg:main}) is elementary; see \cref{sec:portals}.
A dynamic program to compute a shortest $r$-basic solution (Line 6 of \cref{alg:main}) is an adaptation from \cite{KNW21}; see \cref{sec:DP}.

This completes the proof of \cref{thm:main}.

\section{Proof of the Structure Theorem in Two Dimensions}

\label{sec:structure}

We define $\AC(\pi,C)$ to be a \emph{subset} of the crossings between a tour $\pi$ and the boundary of a cell $C$ in \cref{def:A}.
The reason for which we consider a subset of the crossings instead of the entire set of the crossings is because it  enables us not to worry about patching the tour along segments that cross the cell $C$ without visiting any points from $P$.

\begin{definition}
\label{def:A}
For a tour $\pi$ and a cell $C$, let $I_0$ denote the set of curves obtained by restricting $\pi$ to the interior of $C$.
Let $I\subseteq I_0$ be the set of curves $Q\in I_0$ such that $Q$ visits \emph{at least one} input point from $P$.
We define $\AC(\pi,C)$ to be the (multi-)set of points consisting of the two endpoints of all curves in $I$.
\end{definition}

\begin{definition}[$r$-basic]\label{def:r-basic}
We say that a tour $\pi$ is \emph{$r$-basic} if (i) for every side $F$ of a cell $C$ in $\CQ$, we have $\AC(\pi,C)\cap F\subseteq Z(F,|\AC(\pi,C)\cap F|)$ and (ii) every element in $\AC(\pi,C)$ has multiplicity at most 2.
\end{definition}

Let $T$ be a 2-approximate solution.
Let $S=\OPT\cup T$.

For simplicity, we assume that $h$ is a horizontal line, and $x_1, x_2, \dots, x_k$ are the $k$ crossings between $S$ and $h$ in increasing order of their $x$-coordinates.
We define the \emph{proximity} of $x_j$ as $\pro_S(x_j)=\dist(x_j, x_{j-1})$ for $j>1$ and we let $\pro_S(x_1)=\infty$.

Let $C$ be a cell.
Let $F$ be a side of $C$.
For any crossing $x$ between $S$ and $F$, we say that $x$ is \emph{active}  if $x \in \AC(\OPT,C)$, and \emph{inactive} otherwise.
Intuitively, active crossings are the points that need to be ``moved'' so as to modify $\OPT$ into  an $r$-basic solution $\OPT'$; inactive crossings are not moved.\footnote{Inactive crossings are useful to create connections when moving active crossings in Step~3 of the construction as well as to bound the total cost of moving active crossings in the analysis.
The distinction between active crossings and inactive crossings is a purely technical consideration.}

\paragraph{Construction of $\OPT'$.}
Let $\OPT'$ be initialized to $\OPT$.
For each cell $C$ in the top-down order and for each side $F$ of $C$, we modify $\OPT'$ using the following steps:
\begin{enumerate}
    \item Let $N_S$ be the set of ``near'' crossings, that is, $N_S$ is the set of crossings $x$ between $S$ and $F$ satisfying $\pro_S(x)\leq \frac{L}{2^i\cdot r}$, where $i$ is the level of $F$.
    \item Let $G_S$ be the set of remaining crossings between $S$ with $F$.
    \item Create a set of line segments $\PF_F$ by connecting each vertex from $N_S$ to its successor.\footnote{The notation $\PF_F$ follows from  \cite{KNW21}. See Figure~2 in the full version of~\cite{KNW21} for an illustration on $\PF_F$.}
    \item For each line segment $\eta\in\PF_F$ that contains at least one active crossing:
    \begin{itemize}
        \item Let $x$ denote  the (only) point on $\eta$ that belongs to $G_S$;
        \item Connect $x$ to the closest point in $Z(F,|G_S|)$;
    \end{itemize}
     \item Apply the patching lemma of Arora (\cref{lem:Arora-patch}) to each line segment $\eta\in\PF_F$ that contains at least one active crossing, so that $F$ is crossed by active crossings only at points of $Z(F,|G_S|)$, and at most twice at each of these points.
\end{enumerate}
This completes the construction of $\OPT'$.


\begin{lemma}
\label{lem:feasible}
    $\OPT'$ is an $r$-basic solution.
\end{lemma}
\begin{proof}
Consider any side $F$ of any cell $C$ such that $\AC(\OPT,C)\cap F\neq \emptyset$.
If $|G_S|=1$, then $\AC(\OPT',C)$ contains a single element on $F$, and that element belongs to $Z(F,1)$. The claim follows.
If $|G_S|\geq 2$, let $G'\subseteq G_S$ denote the set of points in $G_S$ whose corresponding line segment contains at least one active crossing in $\OPT$.
Then $|\AC(\OPT',C)\cap F|=|G'|\leq |G_S|$.
Since $g(\cdot)$ is a non-increasing function and the values of $g(\cdot)$ are powers of 2, we have $g(|G'|)$ is an integer multiple of $g(|G_S|)$.
Therefore, $Z(F,|G'|)\supseteq Z(F,|G_S|)$.
By construction, each point in $\AC(\OPT',C)\cap F$ belongs to $Z(F,|G_S|)$, thus belongs to $Z(F,|G'|)$.
Since  $|\AC(\OPT',C)\cap F|=|G'|$, the claim follows.
\end{proof}


\begin{lemma}
\label{lem:cost-analysis}
    The expected total connection cost in the construction of $\OPT'$ is $O(1/r)\cdot\opt$.
\end{lemma}

In the rest of the section, we show \cref{lem:cost-analysis}.

\subsection{Multiple Crossings With a Side: Previous Analysis}
\label{sec:multiple-crossings}

\begin{lemma}[adaptation from \cite{KNW21}]
\label{lem:patching-KNW}
    The expected total connection cost in Step~3 of the construction is $O(1/r)\cdot \opt$.
    The expected total connection cost over all sides $F$ of all cells $C$ such that $|G_S|\geq 2$ in Step~4 of the construction is $O(1/r)\cdot \opt$.
\end{lemma}

\begin{proof}
    The proof is identical to the proof of Theorem III.3 in \cite{KNW21}, except that we replace $\OPT$ in \cite{KNW21} by $S=\OPT\cup T$.
    Note that $Z(F,|G_S|)\supseteq \grid(F,g(|G_S|))$.
    Thus by \cite{KNW21}, each of the two expected costs in the claim is $O(1/r)\cdot \cost(S)=O(1/r)\cdot (\cost(\OPT)+\cost(T))$, which is at most $O(1/r)\cdot 3\cdot\opt$, since $T$ is a 2-approximate solution.
\end{proof}

\subsection{Single Crossing With a Side: New Analysis}

In this section, we bound the total connection cost in Step 4 of the construction for all sides $F$ of all cells $C$ such that $|G_S|=1$.

Let $F$ be a side of a cell $C$ such that $|G_S|=1$.
Let $x$ be the only element in $G_S$.
There are two cases of $F$, depending on whether $F$ is $T$-crossing.

\subsubsection{Case 1: $F$ is $T$-crossing}
The cost of connecting $x$ to the closest point in $Z(F,1)$ is at most the cost of connecting $x$ to $x_F$, since $x_F\in Z(F,1)$.
Since $|G_S|=1$, $x$ and $x_F$ are considered as ``near'' crossings on $F$ in the construction. By \cref{lem:patching-KNW}, the expected overall cost of connecting ``near'' crossings (i.e., Step 3 of the construction) is $O(1/r)\cdot\opt$, hence the following bound.

\begin{lemma}
\label{lem:single-crossing-and-T-crossing}
The expected total connection cost over all sides $F$ of all cells $C$ such that $|G_S|=1$ and $F$ is $T$-crossing in Step~4 of the construction is $O(1/r)\cdot \opt$.
\end{lemma}

\subsubsection{Case 2: $F$ is not $T$-crossing}
This is the most interesting part of the analysis that contains several new ideas.
\begin{definition}[interesting sides]
Let $F$ be a side of a cell.
We say that $F$ is an \emph{interesting side} if $F$ is not $T$-crossing and is such that $|G_S|=1$.
\end{definition}


\begin{lemma}
\label{lem:interesting}
The expected total connection cost over all interesting sides $F$ of all cells $C$ in Step~4 of the construction is $O(1/r)\cdot \opt$.
\end{lemma}

The key to the proof of \cref{lem:interesting} is the following lemma (\cref{lem:q(e)}).
\cref{lem:q(e)} is a main novelty in this paper.

Let $\PD_i:=\frac{L}{2^{i}\cdot g(1)}$. Thus $\PD_i$ is the interportal distance on a side $F$ at level $i$ where $|G_S|=1$; in other words, $\PD_i$ is the interportal distance in $\grid(F, g(1))$.

\begin{lemma}[main lemma]
\label{lem:q(e)}
For any edge $e\in \OPT$, let $i_e\in \mathbb{N}$ denote the smallest level of a line that intersects $e$.
Let $X\subseteq \OPT$ denote the set of edges $e\in \OPT$ such that $e$ crosses an interesting side at level $i_e$.
Then we have
\[\E_{a,b}\left[\sum_{e\in X} P_{i_e}\right]=O(1/r)\cdot\opt.\]
\end{lemma}

The proof of \cref{lem:q(e)} is in \cref{sec:proof-q(e)}.
The proof of \cref{lem:interesting} is in \cref{sec:proof-interesting}.

\subsection{Proof of \cref{lem:cost-analysis} Using \cref{lem:interesting}}
The claim in \cref{lem:cost-analysis} follows from \cref{lem:patching-KNW,lem:single-crossing-and-T-crossing,lem:interesting}.

\subsection{Proof of the Main Lemma (\cref{lem:q(e)})}
\label{sec:proof-q(e)}

Let $a\in\{1,\dots,L\}$ denote the random shift that decides the levels of vertical lines.
Let $b\in\{1,\dots,L\}$ denote the random shift that decides the levels of horizontal lines.

\begin{definition}[badly cut]
\label{def:badly_cut}
    We say that an edge $e\in X$ is \emph{badly cut} if $\cost(e)<L/(2^{i_e} \cdot r)$.
Let $X_{\rm bad}\subseteq X$ denote the set of edges $e\in X$ that are badly cut.
\end{definition}


\begin{fact}
\label{fact:non-bad}
    \[\sum_{e\in X\setminus X_{\rm bad}} P_{i_e}\leq O(1/r)\cdot\opt.\]
\end{fact}
\begin{proof}
    For any edge $e\in X\setminus X_{\rm bad}$, we have $P_{i_e}=\frac{L}{2^{i_e}\cdot g(1)}\leq O(1/r)\cdot \cost(e)$, since $g(1)=O(r^2)$.
Thus \[\sum_{e\in X\setminus X_{\rm bad}} P_{i_e}\leq O(1/r)\cdot \sum_{e\in X\setminus X_{\rm bad}}\cost(e)\leq O(1/r)\cdot\opt,\]
where the last inequality is because $X\setminus X_{\rm bad}\subseteq \OPT$.
\end{proof}

From \cref{fact:non-bad}, it suffices to consider badly cut edges $e\in X_{\rm bad}$ in the rest of the analysis.

Let $e$ denote an edge in $X_{\rm bad}$.
Since $e\in X$, $e$ crosses a vertical line of level $i_e$ or a horizontal line of level $i_e$.
To simplify the presentation, we assume that $e$ crosses a vertical line of level $i_e$; let $\ell$ denote this line.
Let $C$ denote the unique cell whose vertical midline segment, letting it be $F_e$, is such that $F_e$ belongs to $\ell$ and $F_e$ intersects $e$.
Let $F_e'$ denote the horizontal midline of $C$.
Let $C_1$ denote the half-cell of $C$ that is to the left of $F_e$.
Let $C_2$ denote the half-cell of $C$ that is to the right of $F_e$.
Let $u$ (resp.\ $v$) be the endpoint of $e$ that belongs to $C_1$ (resp.\ $C_2$).
See~\cref{fig:two cases}.

\begin{figure}[H]
    \begin{center}
        \includegraphics[scale=0.16]{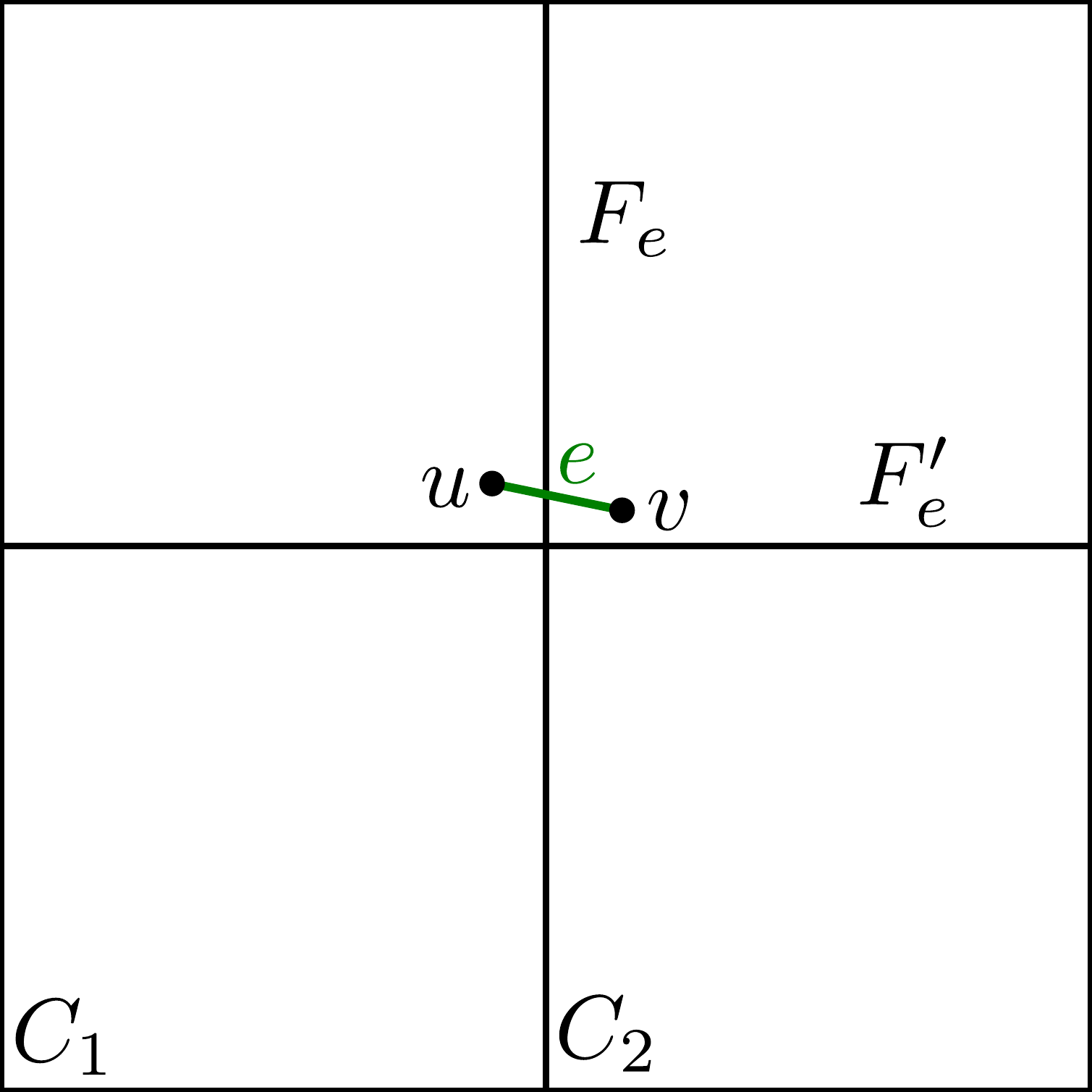}
    \end{center}
    \caption{The edge $e$ crosses the side $F_e$. Here $C_1$ is the rectangle consisting of the two subcells to the left of $F_e$, and $C_2$ is the rectangle consisting of the two subcells to the right of $F_e$.}
    \label{fig:two cases}
\end{figure}

\begin{definition}[centered]
\label{def:centered}
We say that a point on $F_e$ is \emph{centered} if that point is within distance $\len(F_e)/100$
to the center point of $F_e$.
\end{definition}

Let $\tilde X\subseteq X_{\rm bad}$ denote the set of edges $e\in X_{\rm bad}$ such that the crossing between $e$ and $F_e$ is centered.
Intuitively, since an edge $e$ is centered with probability $1/50$, we have the following lemma (\cref{lem:tilde-X}).

\begin{lemma}
\label{lem:tilde-X}
    If $\mathbb{E}_{a,b}\big[\sum_{e\in \tilde{X}} P_{i_e}\big]=O(1/r)\cdot\opt$, then $\mathbb{E}_{a,b}\big[\sum_{e\in X_{\rm bad}} P_{i_e}\big]=O(1/r)\cdot\opt$.
\end{lemma}

The proof of \cref{lem:tilde-X} is in \cref{sec:centered-crossing}.

From \cref{lem:tilde-X}, it suffices to bound $\mathbb{E}_{a,b}\big[\sum_{e\in \tilde{X}} P_{i_e}\big]$. Thus we focus on centered crossings in the rest of the analysis.

\subsubsection{Critical and Non-critical Edges}
\begin{definition}[critical]
\label{def:critical}
Let $e=(u,v)$ be an edge in $\tilde X$.
Let $\pi_1$ be a path in $T$ that starts at a boundary point in $C_1$, visits $u$, and ends at a boundary point in $C_1$.
Let $\pi_2$ be a path in $T$ that starts at a boundary point in $C_2$, visits $v$, and ends at a boundary point in $C_2$.
Let $\Delta_1$ denote the length of $\pi_1$ minus the distance between the two endpoints of $\pi_1$.
Let $\Delta_2$ denote the length of $\pi_2$ minus the distance between the two endpoints of $\pi_2$.
We say that $e$ is \emph{critical} if $\Delta_1+\Delta_2<r\cdot P_{i_e}$ and is \emph{non critical} otherwise.
Let $\CE\subseteq \tilde X$ denote the set of critical edges $e\in \tilde X$.
Let $\NCE\subseteq \tilde X$ denote the set of non-critical edges $e\in \tilde X$.
\end{definition}

\begin{figure}[H]
\begin{center}
\begin{subfigure}[b]{.32\textwidth}
\centering
\includegraphics[width=.8\textwidth]{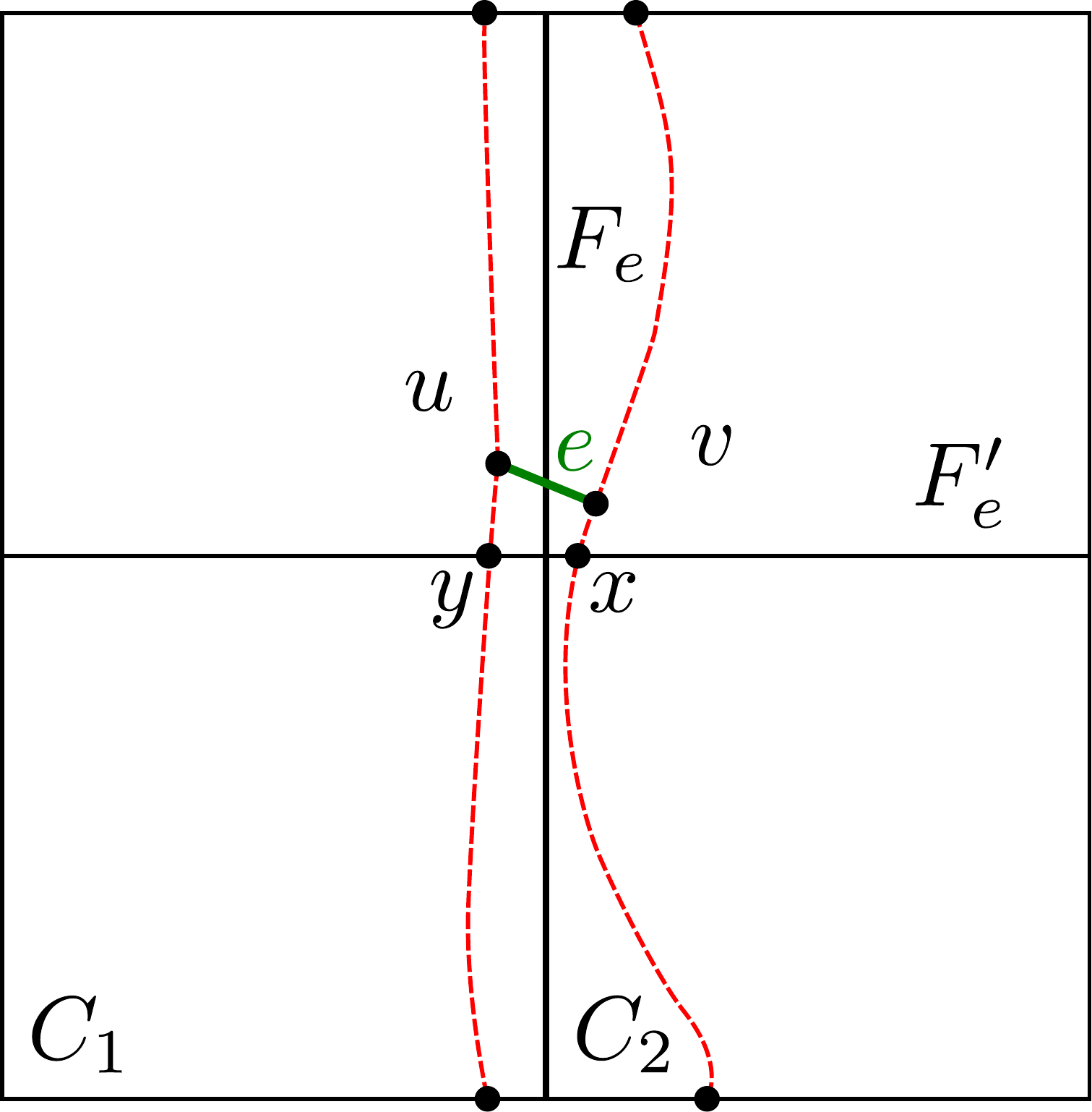}\\
\caption{critical edge $e$}
\label{fig:critical}
\end{subfigure}
\begin{subfigure}[b]{.32\textwidth}
\centering
\includegraphics[width=.8\textwidth]{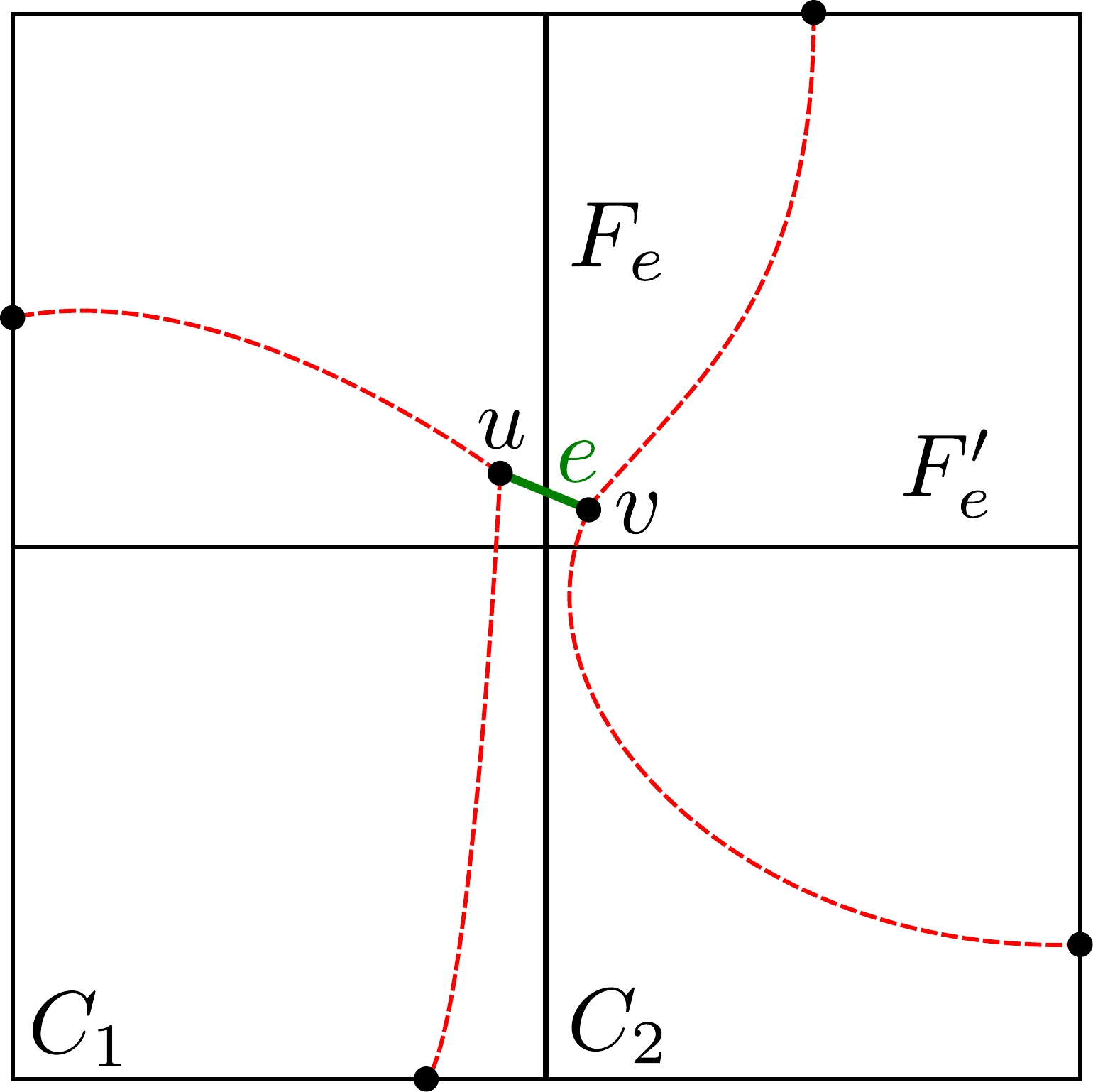}\\
\caption{non critical edge $e$}
\label{fig:non-critical}
\end{subfigure}
\caption{The dashed red lines depict sub-paths of $T$.}
\label{fig:critical-and-non-critical}
\end{center}
\end{figure}

We bound the total cost for critical edges and the total cost for non-critical edges in the next two subsections.

\subsubsection{Bounding the cost for critical edges}
\label{sec:critical}

In order to bound the cost for critical edges, first, we define $\beta_{i,a}(x)$ in \cref{def:beta}. This definition is inspired by the definition of $\alpha_i(x)$ in \cite{KNW21}.


\begin{definition}
\label{def:beta}
For each crossing point $x$ between $S$ and a horizontal line $h$, for each integer $i\in \mathbb{N}$, and for each integer $a\in\{1,\dots,L\}$, we define $\beta_{i,a}(x)$ as follows:\\[2mm]
If $\pro_S(x)> \frac{L}{2^i r}$, then
\[
\beta_{i,a}(x):=8\cdot \left(\frac{L}{2^i r}\right)^2\frac{1}{\pro_S(x)}.
\]
\noindent
If $\pro_S(x)\leq\frac{L}{2^i r}$,  then
\[
    \beta_{i,a}(x):=
    \begin{cases}
    \PD_i, &\text{if $[-a+x^{(1)}-\pro_S(x),-a+x^{(1)}]$ contains an integer multiple of $L/2^i$}\\
    \pro_S(x), &\text{otherwise,}
    \end{cases}
\]
where $x^{(1)}$ denotes the horizontal coordinate of the point $x$.
\end{definition}

Let $e$ be a critical edge. See \cref{fig:critical}.

\begin{fact}
\label{fact:pi}
Both $\pi_1$ and $\pi_2$ connect the top boundary with the bottom boundary.
\end{fact}

\begin{proof}
Since $e\in \tilde X$, the crossing between $e$ and $F_e$ is within distance $\len(F_e)/100$ to the center of $F_e$.
Since $e$ is badly cut, the length of $e$ is at most $\len(F_e)/(2r)$.
Thus both $u$ and $v$ are at distance at most $(1/100+1/(2r))\cdot \len(F_e)$ to the center of $F_e$.
Since $e$ is critical and $P_{i_e}=O(1/r^2)\cdot\len(F_e)$, both $\Delta_1$ and $\Delta_2$ are smaller than $r\cdot  P_{i_e}\leq O(1/r)\cdot \len(F_e)$.

Suppose for the sake of contradiction that $\pi_1$ (resp.\ $\pi_2$) does not connect the top boundary with the bottom boundary.
Since both endpoints of $\pi_1$ (resp.\ $\pi_2$) are on the boundary of $C$ and using the fact that $\pi_1$ (resp.\ $\pi_2$) visits $u$ (resp.\ $v$) which is within distance $(1/100+1/(2r))\cdot \len(F_e)$ to the center of $F_e$, we must have $\Delta_1$ (resp.\ $\Delta_2$) is at least $(1/5)\cdot \len(F_e)$, which contradicts  the fact that $\Delta_1$ (resp.\ $\Delta_2$) is smaller than $O(1/r)\cdot \len(F_e)$ when $r$ is large enough.
The claim follows.
\end{proof}

From \cref{fact:pi}, we immediately have the following structure property (\cref{lem:Fe'}).

\begin{lemma}
\label{lem:Fe'}
    $F_e'$ crosses both $\pi_1$ and $\pi_2$.
\end{lemma}

Let $y$ (resp.\ $x$) denote the crossing between $\pi_1$ (resp.\ $\pi_2$) and $F_e'$; see \cref{fig:critical}.
Since $\pi_1$ and $\pi_2$ are disjoint, $y$ and $x$ are distinct.

\begin{lemma}
\label{lem:beta}
$P_{i_e}\leq \beta_{i_e,a}(x)$.
\end{lemma}

\begin{proof}
Recall that the level of $F_e'$ is $i_e$.
Since both $x$ and $y$ belongs to $F_e'$, we have $\pro_S(x) < \len(F_e')=2\cdot (L/2^{i_e})$.

There are two cases.

Case 1: when
$\pro_S(x) > L/(2^{i_e}\cdot r)$.
Using that $g(1)\geq r^2/4$, we have
\[\beta_{i_e,a}(x)=8\cdot \left(\frac{L}{2^{i_e}\cdot r}\right)^2\frac{1}{\pro_S(x)}\geq \frac{4}{r^2}\cdot\frac{L}{2^{i_e}}\geq \frac{L}{2^{i_e}\cdot g(1)}=P_{i_e}.\]

Case 2: when $\pro_S(x)\leq L/(2^{i_e}\cdot r)$.
Observe that the segment $xy$ is cut by a level $i_e$ vertical line.
This implies that $[-a+x^{(1)}-\pro_S(x),-a+x^{(1)}]$ contains an integer multiple of $L/2^{i_e}$.
Thus by the definition of $\beta_{{i_e},a}(x)$, we have
$\beta_{{i_e},a}(x)=P_{i_e}.$
\end{proof}

\begin{lemma}
\label{lem:CE}
\[\mathbb{E}_{a,b}\left[\sum_{e \in \CE} P_{i_e}\right]\leq  O(1/r)\cdot \opt.\]
\end{lemma}

\begin{proof}
For each $e\in \CE$, by \cref{lem:beta}, $P_{i_e}\leq \beta_{i_e,a}(x)$, where $x$ is defined w.r.t.\ $e$.
We \emph{charge} $P_{i_e}$ to $\beta_{i_e,a}(x)$.
We observe that each $\beta_{i,a}(x)$ is charged by at most one edge $e\in \CE$, because we are in the single crossing case.\footnote{Technically, there can be more than one edge $e\in \CE$ that crosses at that single crossing. But those crossings must be  ``near'' crossings, and are thus combined into a single crossing. So we only need to charge $\beta_{i,a}(x)$ once to pay the connection of moving that single  crossing to the closest portal.}

Consider any crossing $x$ between $S$ and a horizontal line $h$ of level $i$.
For each integer $i\leq \log L$, we analyze $\mathbb{E}_a[\beta_{i,a}(x)]$ as follows.
There are two cases.
When $\pro_S(x)> \frac{L}{2^{i}\cdot r}$, we have
\[
\mathbb{E}_a[\beta_{i,a}(x)]=8\cdot \left(\frac{L}{2^i r}\right)^2\frac{1}{\pro_S(x)}.
\]
When $\pro_S(x)\leq \frac{L}{2^i\cdot r}$, since $a$ is independent of $i$, the probability over $a$ that $[-a+x^{(1)}-\pro_S(x),-a+x^{(1)}]$ contains an integer multiple of $L/2^i$ is $\frac{\pro_S(x)}{L/2^i}$.
Thus
\begin{align*}
\mathbb{E}_a[\beta_{i,a}(x)] &= \frac{\pro_S(x)}{L/2^i}
 \cdot \frac{L}{2^i\cdot g(1)} + \left(1-\frac{\pro_S(x)}{L/2^i}
 \right)\pro_S(x)\\
 & <
\left(1+\frac{1}{g(1)}\right)\cdot \pro_S(x)\\
&\leq
2\cdot \pro_S(x).
\end{align*}

Using a similar argument from \cite{KNW21}, we have
\begin{align*}
\mathbb{E}_{a,b}\left[\beta_{i,a}(x)\right]&\leq \sum_{i=0}^{\log L} \mathbb{P}[h \text{ has level }i]\cdot \mathbb{E}_a[\beta_{i,a}(x)]\\
& =O\left(\sum_{i=0}^{\theta} \frac{2^i}{L} \cdot 2\cdot \pro_S(x) + \sum_{i=\theta+1}^{\log L} \frac{2^i}{L}\cdot 8\cdot \left(\frac{L}{2^i r}\right)^2\frac{1}{\pro_S(x)}\right)\\
& =O(1/r),
\end{align*}
where $\theta:=\log \frac{1}{r\cdot \pro_S(x)}$, and the last equality is by the convergence of sums of geometric progressions.
Since the total number of crossings between $S$ and all horizontal lines is $O(\opt)$, the expected total charge is $O(1/r)\cdot \opt$. The claim follows.
\end{proof}

\subsubsection{Bounding the cost for non-critical edges}

We consider the non-critical edges $e$ in decreasing order of $i_e$.
Let $e$ be any non-critical edge.
See \cref{fig:non-critical}.
We replace $\pi_1$ (resp.\ $\pi_2$) by a straight segment connecting the endpoints of $\pi_1$ (resp.\ $\pi_2$).
The solution is shortened by $\Delta_1+\Delta_2$.
We \emph{charge} $P_{i_e}$ to $\Delta_1+\Delta_2$.
By the definition of non-critical edges, we have \[P_{i_e}\leq (1/r)\cdot(\Delta_1+\Delta_2).\]
Since the charges are disjoint, the overall charge is at most $\cost(T)\leq 2\cdot \opt$.
Thus we obtain the following bound (\cref{fact:NCE}).

\begin{fact}
\label{fact:NCE}
\[\sum_{e \in \NCE} P_{i_e}\leq  O(1/r)\cdot \opt.\]
\end{fact}

\subsubsection{Reducing to Centered Crossings}
\label{sec:centered-crossing}
In this section, we prove \cref{lem:tilde-X}.

Consider an edge $e\in \OPT$.
Fix $a\in \{1,\dots, L\}$ which decides the level of vertical lines.
Let $h$ denote a vertical line of the smallest level that crosses $e$.
Let $i$ denote the level of $h$.
We assume that $\pro_S(x)>\frac{L}{2^i\cdot r}$, i.e., $x$ is not a ``near'' crossing along $h$.
This is without loss of generality, because ``near'' crossings have already been analyzed in \cref{sec:multiple-crossings}.

For any $b\in \{1,\dots,L\}$ which decides the level of horizontal lines, we define $w_b(x)$, which is an upper bound of the cost for moving $x$ in the construction w.r.t.\ the shift $b$:
\[w_b(x):=
\begin{cases}
    \PD_i, & \text{if $x$ is moved to a grid point (case 1);} \\
    \dist(x,y), & \text{if $x$ is moved to a point $y\in T$ (case 2)}.
\end{cases}
\]
We assume that, in case 2, $\dist(x,y)\leq \PD_i$, since otherwise we may move $x$ to the closest portal and switch $x$ to case 1 without increasing $w_b(x)$.

\begin{lemma}
\label{lem:centered}
Fix $a\in \{1,\dots, L\}$.
For any $b\in \{1,\dots, L\}$, let $i_b(x)\in \{0,1\}$ be an indicator such that $i_b(x):=1$ if $x$ is a crossing between $e$ and an interesting side, and $i_b(x):=0$ otherwise.
Let $c_b(x)\in \{0,1\}$ be an indicator such that $c_b(x):=1$ if $x$ is centered on $F_e$, and $c_b(x):=0$ otherwise.
We have \[\displaystyle\E_b[i_b(x)\cdot \PD_i]\leq 50\cdot \mathbb{E}_b[c_b(x)\cdot w_b(x)].\]
\end{lemma}

\begin{proof}
First, we show that, for any shift  $b_0$ such that $x$ is centered w.r.t.\ $b_0$, we have $\E_b[i_b(x)]\cdot \PD_i\leq w_{b_0}(x)$.
If $w_{b_0}(x)=\PD_i$, then the claim follows trivially. Next, consider the case when $w_{b_0}(x)=\dist(x,y)$ for some point $y\in T$.
Consider any shift $b$ such that $x$ is an interesting crossing with $b$.
Then $F(b)$ is not $T$-crossing.
Thus $x$ and $y$ must be separated by  a horizontal line of level $i$.
The probability over $b$ that $x$ and $y$ are separated by a horizontal line of level $i$ is $\dist(x,y)/(L/2^i)$.
Thus we have
\[\mathbb{E}_b[i_b(x)]\cdot\PD_i\leq \frac{\dist(x,y)}{L/2^i}\cdot \frac{L}{2^i\cdot g(1)}=\frac{w_{b_0}(x)}{g(1)}<w_{b_0}(x).\]

Since $\mathbb{P}_b[c_b(x)=1]=1/50$, we conclude that
\[\displaystyle\E_b[i_b(x)]\cdot \PD_i\leq \E_{b\mid c_b(x)=1}[w_{b}(x)]=50\cdot \mathbb{E}_b[c_b(x)\cdot w_b(x)].\]
\end{proof}

By the assumption in \cref{lem:tilde-X},
$\mathbb{E}_{a,b}\big[\sum_{e\in \tilde{X}} P_{i_e}\big]=O(1/r)\cdot\opt$.
The overall cost of moving a point $x$ to a point $y\in T$ is at most $O(1/r)\cdot\opt$ (\cref{lem:single-crossing-and-T-crossing}).
Thus \[\mathbb{E}_{a,b}\left[\sum_x c_b(x)\cdot w_b(x)\right]=O(1/r)\cdot \opt.\]
So by \cref{lem:centered},
\[\displaystyle\E_{a,b}\left[\sum_x i_b(x)\cdot \PD_i\right]\leq 50\cdot \mathbb{E}_{a,b}\left[\sum_x c_b(x)\cdot w_b(x)\right]=(1/r)\cdot\opt.\]
Since every edge $e\in X_{\rm bad}\subseteq X$ has a crossing $x$ with an interesting side of level $i_e$, the claim follows.

\subsection{Proof of \cref{lem:interesting}}
\label{sec:proof-interesting}
Consider any edge $e\in \OPT$.
Let $\mathcal{F}_e$ denote the set of all interesting sides $F$ of all cells $C$ such that $F$ has an \emph{active} crossing with $e$.
Let $W_e$ denote the total connection cost over all sides $F\in \mathcal{F}_e$ in Step~4 of the construction.
It suffices to show that $\sum_{e\in \OPT} W_e=O(1/r)\cdot \opt$.

Observe that every $F\in \mathcal{F}_e$ is of level at least $i_e$, by the definition of $i_e$.
For each level $i\geq i_e$, $\mathcal{F}_e$ contains at most two sides of level $i$.
This is because, any cell $C$ whose boundary has an active crossing with $e$ is such that $C$ contains exactly one of $u$ and $v$.
At each level $i\geq i_e$, there is at most one cell that contains exactly one of $u$ and $v$.
Thus $\mathcal{F}_e$ contains at most two sides of level $i$.

For simplicity, we assume that $e\in \OPT$ crosses a vertical side of level $i_e$, where $F_e$ denote this vertical side.
The analysis for $e\in \OPT$ that crosses a horizontal side of level $i_e$ is analogous.

We decompose $W_e$ into the sum of the following three parts:
\begin{itemize}
    \item $W_e^{(0)}$: the connection cost at $F_e$ if $F_e\in \mathcal{F}_e$; and 0 otherwise.
    \item $W_e^{(1)}$: the total connection cost at all sides $F\in \mathcal{F}_e$ of level at least $\log \frac{L}{r\cdot \cost(e)}$.
    \item $W_e^{(2)}$: the total connection cost at all sides $F\in \mathcal{F}_e$ of level less than $\log \frac{L}{r\cdot \cost(e)}$ s.t.\  $F\neq F_e$.
\end{itemize}

First, observe that $W_e^{(0)}\leq P_{i_e}$ if $F_e$ is interesting; and $W_e^{(0)}=0$ otherwise.

Next, we analyze $W_e^{(1)}$. Since for each level $i$, $\mathcal{F}_e$ contains at most 2 sides of level $i$, we have
\[W_e^{(1)}\leq 2\cdot \sum_{i \geq \log \frac{L}{r \cdot \cost(e)}}\frac{L}{2^i\cdot g(1)}<\frac{4r \cdot \cost(e)}{g(1)} \leq \frac{16 \cdot \cost(e)}{r},\]
where the last inequality follows from $g(1)\geq r^2/4$.

Finally, we analyze $W_e^{(2)}$.
Consider a side $F\in \mathcal{F}_e$ such that the connection cost at $F$ contributes to $W_e^{(2)}$.
We assume that $F$ is not a segment on $F_e$ (otherwise, the crossing between $F$ and $e$ would be the same point as the crossing between $F_e$ and $e$).
Since $F$ has a level of less than $\log \frac{L}{r\cdot \cost(e)}$, we have $\len(F)>r\cdot \cost(e)>2\cdot\cost(e)$, assuming $r>2$.
Using all of the facts that $F_e$ intersects $e$, $F$ intersects $e$, $\len(F)\geq 2\cdot\cost(e)$, $F$ is not a segment on $F_e$, both $F$ and $F_e$ are sides in the dissection, we conclude that $F$ must be a horizontal side whose right endpoint is on $F_e$, see \cref{fig:interesting}.
Let $A$ denote the right endpoint of $F$.
Let $B$ denote the intersection point between $e$ and $F$.
The distance between $A$ and $B$ is at most $\cost(e)$.
Since $A\in Z(F,1)$, the distance between $B$ and the closest point in $Z(F,1)$ is at most $\cost(e)$.
For each level $i<\log\frac{L}{r\cdot\cost(e)}$, the probability that $e$ cuts a horizontal side of level $i$ is at most $\frac{\cost(e)}{L/2^i}$.
Thus \[W_e^{(2)}\leq 2\cdot \cost(e)\cdot \sum_{i<\log\frac{L}{r\cdot\cost(e)}} \frac{\cost(e)}{L/2^i}<\frac{4\cdot \cost(e)}{r}.\]
\begin{figure}[H]
    \begin{center}
        \includegraphics[width=0.4\textwidth]{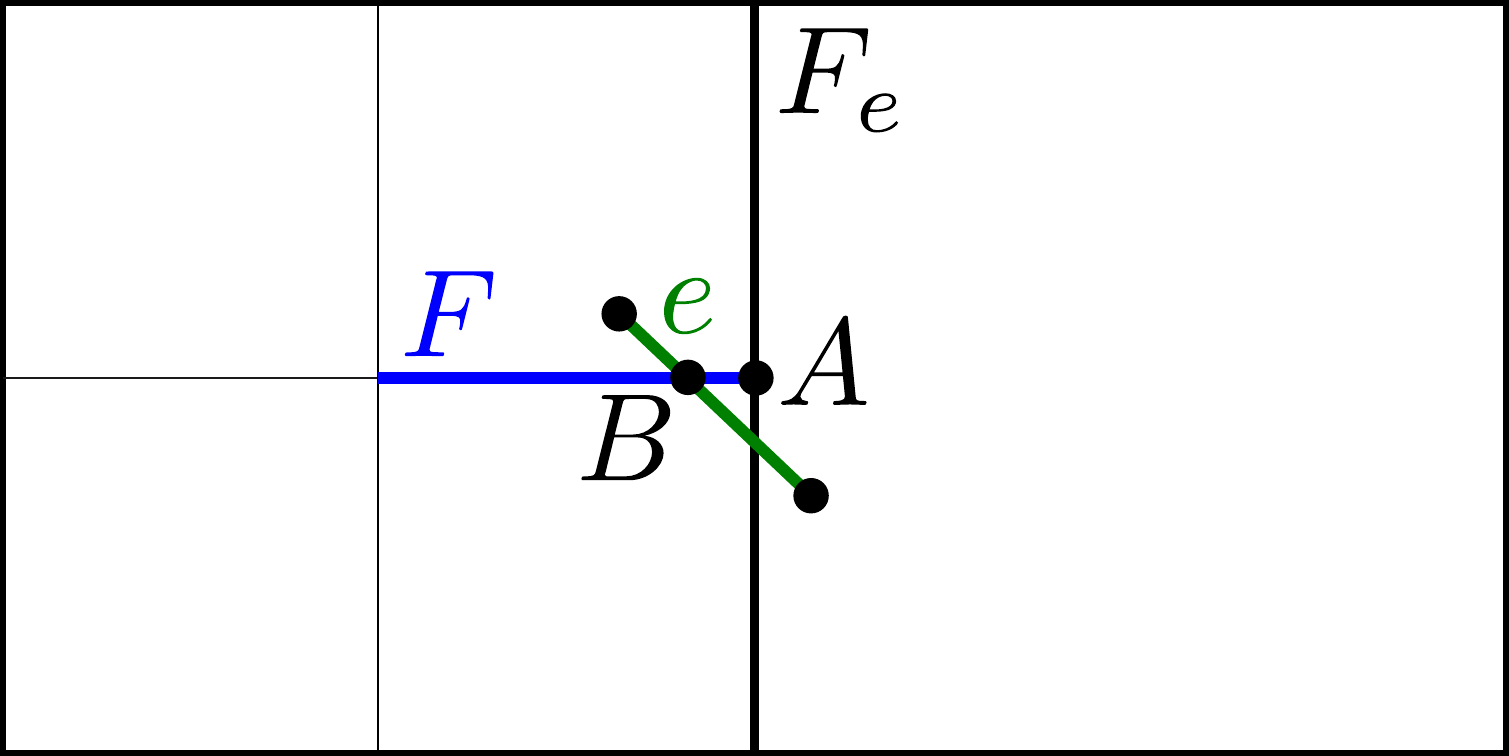}
    \end{center}
    \caption{\label{fig:interesting}The edge $e$ is the segment in green. $F_e$ is the vertical segment in bold. $F$ is the horizontal segment in bold. $A$ is the intersection point between $F_e$ and $F$.
    $B$ is the intersection point between $e$ and $F$. Then $\dist(A,B)\leq \cost(e)$.}
\end{figure}

Summing over all edges $e\in \OPT$, the overall connection cost in the claim is at most
\[\sum_{e\in \OPT} W_e^{(0)}+W_e^{(1)}+W_e^{(2)}
\leq \sum_{e\in X} P_{i_e}+ \sum_{e\in \OPT} (20/r)\cdot\cost(e),\]
where $X\subseteq \OPT$ denotes the set of edges $e\in \OPT$ such that $e$ crosses an interesting side at level $i_e$.
By \cref{lem:q(e)}, \[\mathbb{E}_{a,b}\left[\sum_{e\in X}  P_{i_e}\right]=O(1/r)\cdot \opt.\]
The claim follows.

\subsection{Proof of the Structure Theorem (\cref{thm:structure})}
The Structure Theorem follows directly from \cref{lem:feasible,lem:cost-analysis}.

\section{Proof of the Structure Theorem in $d$ Dimensions}

\label{sec:d-dimensions}

In this section, we prove the Structure Theorem in $d$ dimensions, for any constant $d\geq 3$.
To simplify the presentation, first, we focus on the case when $d=3$. The generalization to $d\geq 4$ is described  at the end of the section.

Most of the analysis in \cref{sec:structure} extends directly to $d$ dimensions, except for \cref{lem:Fe'}, which is used for bounding the cost for critical edges.
The key in this section is \cref{lem:3d-2crossings}, which is a generalization of \cref{lem:Fe'} to 3 dimensions.


Let $e$ be a critical edge.
Let $i_e$ denote the smallest level of a plane that intersects $e$.
Let $\ell$ denote such a plane.
Let $C$ denote the unique cell that has a mid-square that belongs to $\ell$ and intersects $e$.
Let $F_e$ denote that mid-square of $C$.
To simplify the presentation, we assume that $C$ is the unit cube, i.e., $C=[0,1]^3$.
Let $z$ denote the crossing between $e$ and $F_e$.
See \cref{fig:mid}.
Let $C_1$ and $C_2$ denote the two half-cells of $C$ separated by $F_e$.
We define two curves $\pi_1$ and $\pi_2$ as in \cref{sec:structure}.
As in \cref{sec:structure}, we say that a point $x\in F_e$ is \emph{centered} if $x$ is within distance $1/100$ to the center point of $F_e$. We assume without loss of generality that $z$ is centered.

\begin{figure}[h]
\centering
\includegraphics[width=.22\textwidth]{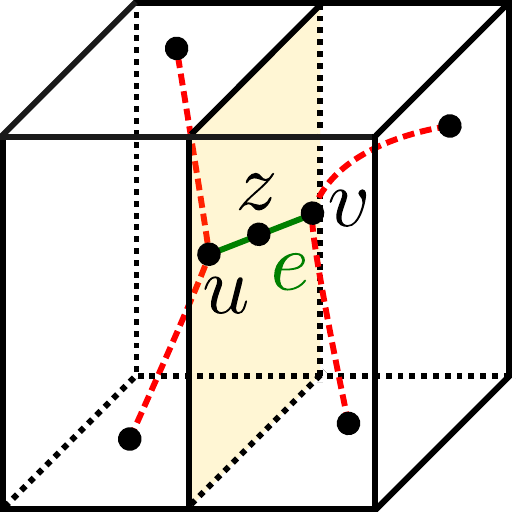}
\caption{The cell $C$ is represented by the cube. $F_e$ is the highlighted square, which is a mid-square of $C$. The edge $e=(u,v)$ is in green. $\pi_1$ and $\pi_2$ are the red dashed curves. $z$ is the crossing point between $e$ and $F_e$.}
\label{fig:mid}
\end{figure}

The main idea in this section is the following lemma (\cref{lem:3d-2crossings}).

\begin{lemma}
\label{lem:3d-2crossings}
    Consider three rectangles that are perpendicular to $F_e$: the first one through the horizontal midline of $F_e$ (\cref{fig:3d-a}), the second one through the vertical midline of $F_e$ (\cref{fig:3d-b}), and the third one through a diagonal of $F_e$ (\cref{fig:3d-c}).
    Then at least one of the three rectangles crosses both $\pi_1$ and $\pi_2$.
\end{lemma}

\begin{figure}[h]
\centering
\begin{subfigure}{.25\textwidth}
\centering
\includegraphics[width=.8\textwidth]{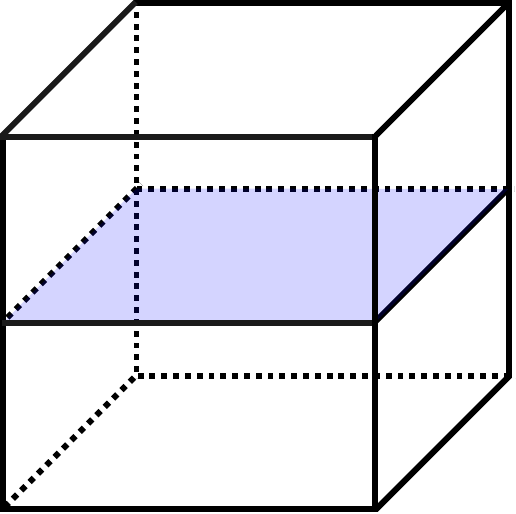}
\caption{}
\label{fig:3d-a}
\end{subfigure}
\begin{subfigure}{.25\textwidth}
\centering
\includegraphics[width=.8\textwidth]{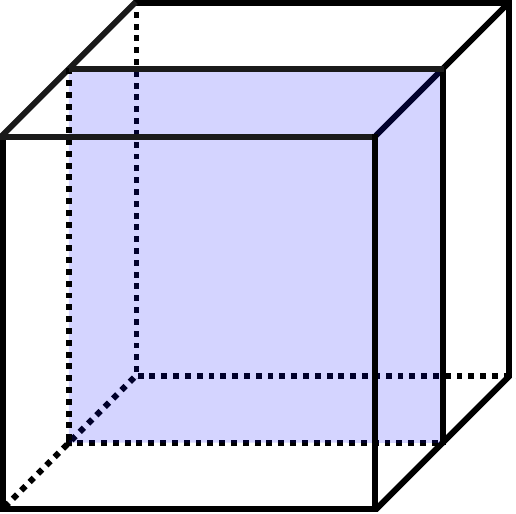}
\caption{}
\label{fig:3d-b}
\end{subfigure}
\begin{subfigure}{.25\textwidth}
\centering
\includegraphics[width=.8\textwidth]{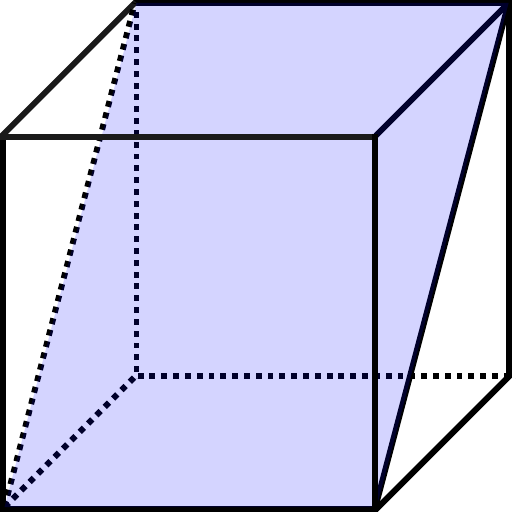}
\caption{}
\label{fig:3d-c}
\end{subfigure}
\caption{\label{fig:projections} Three rectangles in \cref{lem:3d-2crossings}.}
\end{figure}

The proof of the Structure Theorem for $d=3$ follows from \cref{lem:3d-2crossings} and a straightforward adaptation from \cref{sec:structure}.

To prove \cref{lem:3d-2crossings}, it suffices to prove the following lemma  (\cref{lem:2d-2crossings}).

\begin{lemma}
\label{lem:2d-2crossings}
Consider three segments in $F_e$: the horizontal midline segment (\cref{fig:horizontal_0}), the vertical midline segment (\cref{fig:vertical_0}), and a diagonal line segment (\cref{fig:diagonal_0}).
For each $i\in\{1,2\}$, let $\pi_i'$  denote the projection of $\pi_i$ onto $F_e$.
Then at least one of the three segments crosses both $\pi_1'$ and $\pi_2'$.
\end{lemma}

\begin{figure}[h]
\centering
\begin{subfigure}{.25\textwidth}
\centering
\includegraphics[width=.8\textwidth]{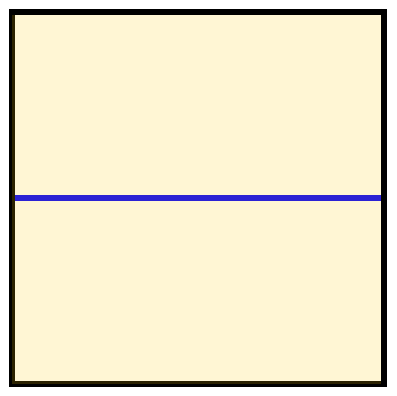}
\caption{}
\label{fig:horizontal_0}
\end{subfigure}
\begin{subfigure}{.25\textwidth}
\centering
\includegraphics[width=.8\textwidth]{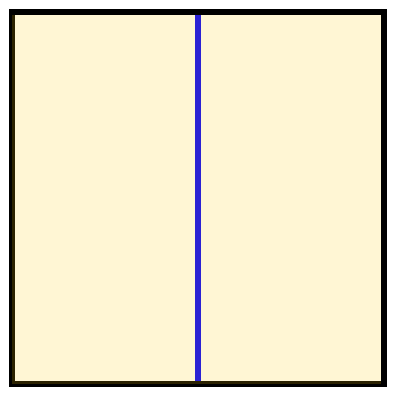}
\caption{}
\label{fig:vertical_0}
\end{subfigure}
\begin{subfigure}{.25\textwidth}
\centering
\includegraphics[width=.8\textwidth]{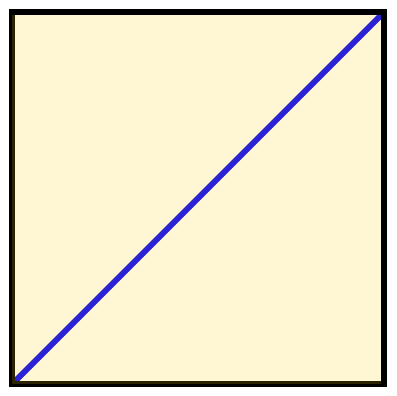}
\caption{}
\label{fig:diagonal_0}
\end{subfigure}
\caption{Three segments in $F_e$ in the claim of \cref{lem:2d-2crossings}. $F_e$ is the highlighted square.}
\end{figure}

In the rest of the section, we prove \cref{lem:2d-2crossings}.

\begin{fact}
\label{fact:pi3d}
Let $\pi\in \{\pi_1, \pi_2\}$.
Let $p$ and $q$ denote the projections of the two endpoints of $\pi$ onto $F_e$.
Let $\tilde p$ denote the point that is symmetric to $p$ with respect to the center $O$ of $F$.
See \cref{fig:pi3d}.
Then both $p$ and $q$ are on the boundary of $F_e$, and in addition, $\dist(q,\tilde p)\leq 0.1$.
\end{fact}

\begin{figure}[h]
\centering
\includegraphics[width=.25\textwidth]{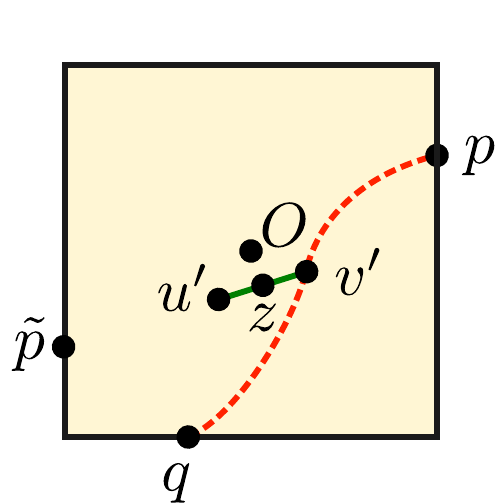}
\caption{$F_e$ is the highlighted square. The point $O$ is the center of the $F_e$. The red curve is the projection $\pi'$ of $\pi$ onto $F_e$.}
\label{fig:pi3d}
\end{figure}

\begin{proof}[Proof Sketch.]
Let $z$ be the crossing between $e$ and $F_e$.
Since $z$ is centered, $\dist(z,O)\leq 1/100$.
Let $u'$ and $v'$ be the projections of $u$ and $v$ onto $F_e$.
Since $e=(u,v)$ is a badly cut edge, $\dist(u,v)<1/r$, thus $\dist(z,v')<\dist(u',v')\leq \dist(u,v)<1/r$.
So $\dist(O,v')\leq \dist(O,z)+\dist(z,v')<(1/100) + (1/r)$.
Let $\pi'$ denote the projection of $\pi$ onto $F_e$.
Let $\Delta$ denote the length of $\pi'$ minus $\dist(p,q)$.
Since $e$ is a critical edge, $\Delta=O(1/r)$.
Combining the above facts with an elementary computation concludes the proof.
\end{proof}

\begin{figure}[h]
\centering
\begin{subfigure}{.25\textwidth}
\centering
\includegraphics[width=.8\textwidth]{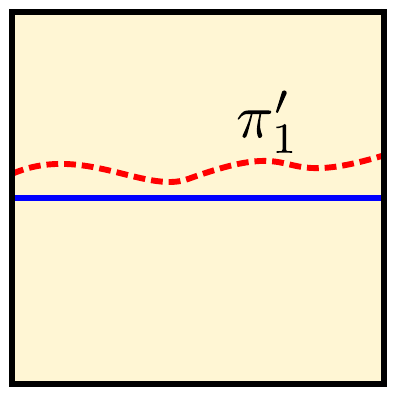}
\caption{}
\label{fig:horizontal}
\end{subfigure}
\begin{subfigure}{.25\textwidth}
\centering
\includegraphics[width=.8\textwidth]{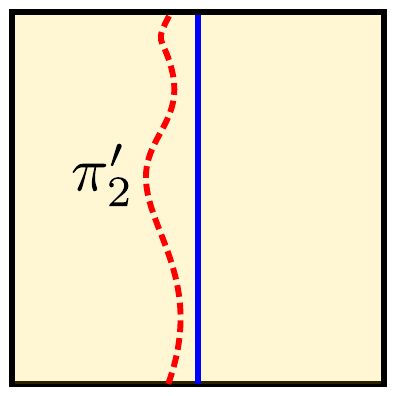}
\caption{}
\label{fig:vertical}
\end{subfigure}
\begin{subfigure}{.25\textwidth}
\centering
\includegraphics[width=.8\textwidth]{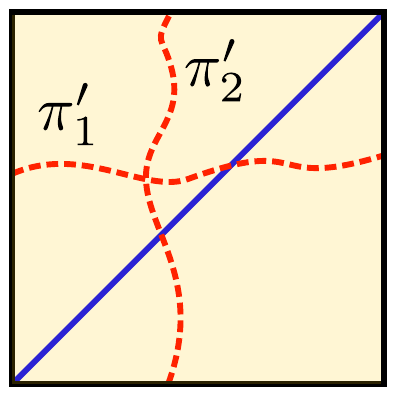}
\caption{}
\label{fig:diagonal}
\end{subfigure}
\caption{Illustration for the proof of \cref{lem:2d-2crossings}. $F_e$ is the highlighted square.}
\end{figure}

\begin{proof}[Proof of \cref{lem:2d-2crossings}]
If both $\pi_1'$ and $\pi_2'$ cross the horizontal midline, then we are done.
So without loss of generality, we assume that $\pi_1'$ does not cross the horizontal midline. See \cref{fig:horizontal}.
Then by \cref{fact:pi3d}, $\pi_1'$ must connect the left boundary of $F_e$ with the right boundary of $F_e$.
Thus $\pi_1'$ crosses the vertical midline.
If $\pi_2'$ also crosses the vertical midline, then we are done.
It remains to consider the case when $\pi_2'$ does not cross the vertical midline.
By \cref{fact:pi3d}, $\pi_2'$ must connect the top boundary of $F_e$ with the bottom boundary of $F_e$.
See \cref{fig:vertical}.
Thus the diagonal of $F_e$ crosses  both $\pi_1'$ and $\pi_2'$.
See \cref{fig:diagonal}.
The claim follows.
\end{proof}

Finally, consider the case when $d\geq 4$. We observe that, for each $i\in \{1,2\}$, there exists an integer $k_i\in [1,d]$ such that at least one endpoint of $\pi_i$ has the ${(k_i)}^{\rm th}$ coordinate in $\{0,1\}$.
We consider the two dimensional plane $\ell$ spanning the ${(k_1)}^{\rm th}$ dimension and the ${(k_2)}^{\rm th}$ dimension.
We define $F_e$ to be the square that is the projection of the cell $C$ onto $\ell$.
The rest of the analysis is a direct adaptation from the case when $d=3$.

\section{Computing Crossings}
\label{sec:portals}
In this section, we show the following lemma.

\begin{lemma}
    Computing the crossings between $T$ and all facets $F$ of all cells $C$ (Lines 3--5 of \cref{alg:main}) can be done in $O(\sqrt{d}\cdot n/\eps)$ time.
\end{lemma}

We assume that $\cost(T)\leq 2\sqrt{d}\cdot n/\eps$.
This is without loss of generality according to the Rounding Assumption of Rao and Smith~\cite[Section 2]{RS98}.
As a consequence, the total number of crossings between $T$ and all horizontal lines (resp.\ all vertical lines) is $O(\sqrt{d}\cdot n/\eps)$.

Our goal is for every crossing $x$ to know which quadtree cell $C$ has a facet $F$ containing $x$ and for every cell $C$ to know which crossings are on the facets of $C$.
According to Rao and Smith~\cite{RS98}, this ``knowledge'' can be computed in $O(1)$ time per crossing if we have the atomic floor operations.\footnote{See the last paragraph in the proof of Lemma~12 in~\cite{RS98}.}
Since the total number of crossings between $T$ and all horizontal lines (resp.\ all vertical lines) is $O(\sqrt{d}\cdot n/\eps)$,  the overall computation time is $O(\sqrt{d}\cdot n/\eps)$.

\section{Dynamic Program}
\label{sec:DP}

In this section, we show the following lemma.

\begin{lemma}\label{lem:DP-running-time}
There is a dynamic program that computes a shortest $r$-basic solution in  $2^{O(1/\eps^{d-1})} n$ time.
\end{lemma}

Our dynamic program is mostly the same as in \cite{KNW21}.
The main difference is that, when combining solutions from children cells, we might include direct connections, i.e.,\ segments without visiting any point from $P$.

We describe how we include direct connections in the dynamic program of \cite{KNW21}.
Let $F$ be a facet of a cell $C$.
We guess the number $k$ of crossings between our solution and $F$.
For a given $k$, the set of portals is $\grid(F,g(k))$, plus, possibly another portal $x_F$ (recall that $x_F$ is pre-computed in Line~5 of \cref{alg:main}).
Note that crossings at corners of $C$ do not count towards the number $k$ of crossings.
We consider two types of connections between portals.
The first type of connections are paths visiting at least one point in $P$.
They are called \emph{multipath-type} connections.
We deal with those connections similarly as~\cite{KNW21} using the formulation as a multipath problem.
The second type of connections are direct segments without visiting any point in $P$.
They are called \emph{direct-type} connections.
For any non-leaf cell $C$, in order to compute a solution within $C$, we combine multipath-type connections, which are solutions from children cells, with  direct-type connections, which are newly created segments. See \cref{fig:stitching}.

\begin{figure}[tb]
\begin{center}
    \includegraphics[width=40mm]{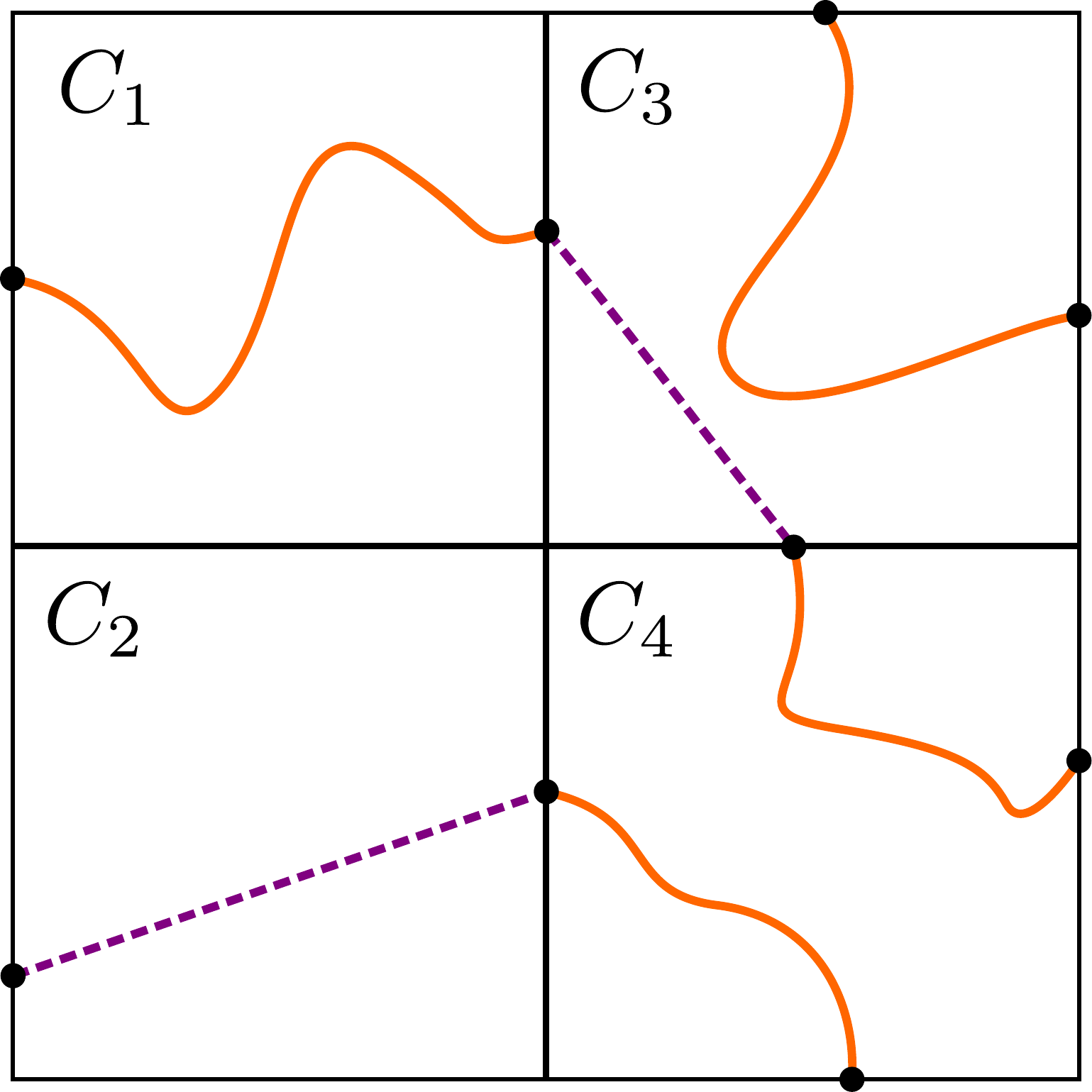}
\caption{\label{fig:stitching} A cell $C$ with four children $C_1, C_2, C_3 ,C_4$. The solid curves are multipath-type connections. They are solutions obtained from children cells. The dashed segments are direct-type connections. They are newly created, i.e., without recursive construction within the children cells of $C$.}
\end{center}
\end{figure}

The rest of the description of the dynamic program and the running time analysis are straightforward adaptations from \cite{KNW21}.

\section{Conclusion}
In this work, we have given a randomized $(1+\eps)$-approximation algorithm for Euclidean TSP in the $d$-dimensional Euclidean space that runs in  $2^{O(1/\eps^{d-1})}n$ time in the real-RAM model with atomic floor or mod operations.
Our approach naturally generalizes to other geometric optimization problems.
For example, combining our approach with the approximation scheme for Steiner Tree from \cite{KNW21} leads to an approximation scheme for Steiner Tree in $2^{O(1/\eps^{d-1})} n$ time.

We leave as an open problem to design a \emph{deterministic} approximation scheme for the Euclidean TSP in $2^{O(1/\eps^{d-1})} n$ time.

\subsection*{Acknowledgment}
We thank Sándor Kisfaludi-Bak for sharing the full version of \cite{KNW21} with us and for answering multiple questions regarding technical details in \cite{KNW21}.

\bibliographystyle{plain}
\bibliography{refs}
\appendix

\end{document}